\documentclass[journal,onecolumn,12pt]{IEEEtran}
\usepackage{pslatex} 
\usepackage{amsfonts,color,morefloats}
\usepackage{amssymb,amsmath,latexsym,amsthm}

\newcommand{\wt}{{\mathrm{wt}}} 

\newcommand{\lcm}{{\rm lcm}}

\newcommand{\tr}{{\mathrm{Tr}}}
\newcommand{\gf}{{\mathrm{GF}}}

\newcommand{\C}{{\mathcal{C}}}

\newcommand{\bc}{{\mathbf{c}}}

\newtheorem{theorem}{Theorem}
\newtheorem{lemma}[theorem]{Lemma}

\newtheorem{example}{Example}

\begin{document}

\title{The Dimension and Minimum Distance of Two Classes of Primitive BCH Codes
\thanks{C. Ding's research was supported by
The Hong Kong Research Grants Council, Proj. No. 16300415. 
C. Fan was supported by the Natural Science Foundation of China, Proj. No. A011601. 
Z. Zhou's research was supported by
the Natural Science Foundation of China, Proj. No. 61201243, and also the Application Fundamental Research Plan Project of Sichuan Province under Grant No. 2013JY0167. }
}

\author{Cunsheng Ding, \and Cuiling Fan, \and Zhengchun Zhou 
\thanks{C. Ding is with the Department of Computer Science
                                                  and Engineering, The Hong Kong University of Science and Technology,
                                                  Clear Water Bay, Kowloon, Hong Kong, China (email: cding@ust.hk).}
\thanks{C. Fan is with the School of Mathematics, Southwest Jiaotong University,
Chengdu, 610031, China (email: fcl@swjtu.edu.cn). }
\thanks{Z. Zhou is with the School of Mathematics, Southwest Jiaotong University,
Chengdu, 610031, China (email: zzc@home.swjtu.edu.cn). }

}

\date{\today}
\maketitle

\begin{abstract} 
Reed-Solomon codes, a type of BCH codes, are widely employed in  communication systems, storage devices and 
consumer electronics. This fact demonstrates the importance of BCH codes -- a family of cyclic codes -- in practice. 
In theory, BCH codes are among the best cyclic codes in terms of their error-correcting capability.  A subclass of  
BCH codes are the narrow-sense primitive BCH codes. However, the dimension and minimum distance of these codes 
are not known in general. The objective of this paper is to determine the dimension and minimum distances of two 
classes of narrow-sense primitive BCH codes with design distances $\delta=(q-1)q^{m-1}-1-q^{\lfloor (m-1)/2\rfloor}$ 
and $\delta=(q-1)q^{m-1}-1-q^{\lfloor (m+1)/2\rfloor}$. The weight distributions of some of these BCH codes are 
also reported. As will be seen, the two classes of BCH codes are sometimes optimal and sometimes among the 
best linear codes known.   
\end{abstract}


\section{Introduction}\label{sec-intro} 

Throughout this paper, let $q$ be a power of a prime $p$. An $[n,k,d]$ linear code $\C$ over $\gf(q)$ is a $k$-dimensional subspace of $\gf(q)^n$ with minimum Hamming distance $d$. 
Let $A_i$ denote the number of codewords with Hamming weight $i$ in a linear code
$\C$ of length $n$. The {\em weight enumerator} of $\C$ is defined by
$$
1+A_1z+A_2z^2+ \cdots + A_nz^n.
$$
The {\em weight distribution} of $\C$ is the sequence $(1,A_1,\ldots,A_n)$.

An $[n,k]$ linear code $\C$ over $\gf(q)$ is called {\em cyclic} if 
$(c_0,c_1, \cdots, c_{n-1}) \in \C$ implies $(c_{n-1}, c_0, c_1, \cdots, c_{n-2})$  
$\in \C$.  
We can identify a vector $(c_0,c_1, \cdots, c_{n-1}) \in \gf(q)^n$ 
with  
$$ 
c_0+c_1x+c_2x^2+ \cdots + c_{n-1}x^{n-1} \in \gf(q)[x]/(x^n-1).  
$$
In this way, a code $\C$ of length $n$ over $\gf(q)$ corresponds to a subset of the quotient ring 
$\gf(q)[x]/(x^n-1)$. 
A linear code $\C$ is cyclic if and only if the corresponding subset in $\gf(q)[x]/(x^n-1)$ 
is an ideal of the ring $\gf(q)[x]/(x^n-1)$. 

It is well-known that every ideal of $\gf(q)[x]/(x^n-1)$ is principal. Let $\C=\langle g(x) \rangle$ be a 
cyclic code, where $g(x)$ is monic and has the smallest degree among all the 
generators of $\C$. Then $g(x)$ is unique and called the {\em generator polynomial,} 
and $h(x)=(x^n-1)/g(x)$ is referred to as the {\em check polynomial} of $\C$. 

From now on, let $m>1$ be a positive integer, and let $n=q^m-1$. 
Let $\alpha$ be a generator of $\gf(q^m)^*$, which is the multiplicative group of $\gf(q^m)$. 
For any $i$ with $0 \leq i \leq q^m-2$, let $m_i(x)$ denote the minimal polynomial of $\alpha^i$ 
over $\gf(q)$. For any $2 \leq \delta < n$, define 
$$ 
g_{(q,m,\delta)}(x)=\lcm(m_{1}(x), m_{2}(x), \cdots, m_{\delta-1}(x)),    
$$
where $\lcm$ denotes the least common multiple of these minimal polynomials. 
We also define 
$$ 
\tilde{g}_{(q,m,\delta)}(x)=(x-1)g_{(q,m,\delta)}(x).  
$$

Let $\C_{(q, m, \delta)}$ and $\tilde{\C}_{(q, m, \delta)}$ denote the cyclic code of length $n$ with generator 
polynomial $g_{(q, m,\delta)}(x)$ and $\tilde{g}_{(q,m,\delta)}(x)$, respectively.  
This set $\C_{(q, m, \delta)}$ is called a \emph{narrow-sense primitive BCH code} with \emph{design distance} $\delta$, 
and $\tilde{\C}_{(q, m, \delta)}$ is called  a \emph{primitive BCH code} with \emph{design distance} $\delta$. 

By definition, the code $\tilde{\C}_{(q, m, \delta)}$ is a subcode of $\C_{(q, m, \delta)}$ and 
$$ 
\dim(\tilde{\C}_{(q, m, \delta)})=\dim(\C_{(q, m, \delta)})-1.
$$

Clearly, $\C_{(q, m, \delta)}$ and $\C_{(q, m, \delta')}$ are identical for two different $\delta$ and $\delta'$, 
as long as $\delta$ and $\delta'$ are from the same $q$-cyclotomic coset modulo $n$.  
The largest design distance of $\C_{(q, m, \delta)}$ is called the \emph{Bose distance} of the code and  
denoted by $d_B$. 
 
The cyclic codes $\C_{(q, m, \delta)}$ are treated  in almost every book on coding theory. However, the following 
questions about the code $\C_{(q, m, \delta)}$ are still open in general. 

\begin{enumerate}
\item What is the dimension of $\C_{(q, m, \delta)}$? 
\item What is the Bose distance (i.e., the maximum design distance) of $\C_{(q, m, \delta)}$? 
\item What is the minimum distance $d$ (i.e., the minimum nonzero weight) of $\C_{(q, m, \delta)}$? 
\end{enumerate} 

The dimension of $\C_{(q, m, \delta)}$ is known 
when $\delta$ is small, and is open in general. There are lower bounds on the dimension of $\C_{(q,m, \delta)}$, which 
are very bad in many cases. The minimum distance $d$ of $\C_{(q, m, \delta)}$ is known only in a few cases. Only 
when $\delta$ is very small or when $\C_{(q,m, \delta)}$ is the Reed-Solomon code, both the dimension and minimum distance 
of $\C_{(q,m, \delta)}$ are known.  Hence, we have very limited knowledge of the narrow-sense primitive BCH codes, not to mention BCH codes in general. Thus, BCH codes are far from being well understood and studied. For known results on BCH codes, the reader is referred to  \cite{Charp90,Charp98,Ding15,DDZ15}. 

In the 1990's, there were a few papers on the primitive BCH codes \cite{ACS92,AS94,Charp90,YH96,YF,YZ}. However, in the 
last sixteen years, little progress on the study of these codes has been made. As pointed out by Charpin in \cite{Charp98}, 
it is a well-known hard problem to determine the minimum distance of primitive BCH codes. 

The objective of this paper is to determine the dimensions and minimum distances of the codes $\C_{(q, m, \delta_i)}$ and 
$\tilde{\C}_{(q, m, \delta_i)}$ with design distances $\delta_2=(q-1)q^{m-1}-1-q^{\lfloor (m-1)/2\rfloor}$ 
and $\delta_3=(q-1)q^{m-1}-1-q^{\lfloor (m+1)/2\rfloor}$. The weight distributions of some of these BCH codes 
are also reported. 

An $[n, k, d]$ linear code is said to be \emph{optimal} if its parameters meet a bound on linear codes, and \emph{almost 
optimal} if the parameters $[n, k+1, d]$ or $[n, k, d+1]$ meet a bound on linear codes.  
To investigate the optimality of the codes studied in this paper, we compare them with the tables of the best linear codes 
known maintained by Markus Grassl at http://www.codetables.de, which is called the {\em Database} later in this paper. 
Sometimes we will employ the tables of best cyclic codes in \cite{Dingbk15} to benchmark some BCH codes dealt with  
in this paper.   

\section{Some auxiliary results} 

The \emph{$q$-cyclotomic coset} modulo $n$ containing $i$ is defined by 
$$ 
C_i=\{ i q^{j} \bmod{n}: 0 \le j < \ell_i\},  
$$
where $\ell_i$ is the smallest positive integer such that $q^{\ell_i} i \equiv i \pmod{n}$, and is called the 
\emph{size} of $C_i$. The smallest integer in $C_i$ is called the \emph{coset leader} of $C_i$. 

It is easily seen that the minimal polynomial $m_i(x)$ of $\alpha^i$ over $\gf(q)$ is given by 
$$ 
m_i(x)=\prod_{j \in C_i} (x - \alpha^j).
$$ 

According to the BCH bound, the minimum distance $d$ of the code $\C_{(q, m, \delta)}$ satisfies 
$$ 
d \ge d_B \ge \delta, 
$$  
where $d_B$ denotes the Bose distance of the code $\C_{(q, m, \delta)}$. 
In some cases the difference $d-\delta$ is very small or zero. In many cases the 
difference $d - \delta$ is very large and in such cases the design distance does not give much information on the 
minimum distance $d$, but the Bose distance $d_B$ may be very close to the minimum distance. In fact, we have 
the following conjecture \cite{Charp98}. 

\vspace{0.3cm}
\noindent 
{\bf Charpin's Conjecture:} The minimum distance $d \le d_B + 4$ for the narrow-sense primitive BCH codes. 

\vspace{0.3cm}
In view of this 
conjecture, it is very valuable to determine the Bose distance for the narrow-sense primitive BCH codes.    

Given a design distance $\delta$, it is a difficult problem to determine the Bose distance $d_B$, not to mention the 
minimum distance $d$. However, we have $d_B = \delta$ if $\delta$ is a coset leader. Therefore, it is imperative  
to choose the design distance to be a coset leader.  In the next two sections, we will choose the design distances of 
several classes of BCH codes in this way.

\section{The parameters of the BCH codes $\C_{(q,\, m,\, \delta_2)}$ and $\tilde{\C}_{(q,\, m,\, \delta_2)}$, 
where $\delta_2=(q-1)q^{m-1}-1-q^{\lfloor (m-1)/2\rfloor}$} 

We need to prove a few lemmas before stating and proving the main results of this section. The following two lemmas are fundamental, and were proved in \cite{Ding15}.  

\begin{lemma}\label{lem-1stcosetleader} 
The largest $q$-cyclotomic coset leader modulo $n=q^m-1$ is $\delta_1=(q-1)q^{m-1}-1.$ Furthermore, $|C_{\delta_1}|=m.$
\end{lemma} 

\begin{theorem}
The code $\tilde{\C}_{(q, m, \delta_1)}$ has papameters $[n,\, m,\, \delta_1+1]$, 
 and meets the Griesmer bound.   
The code $\C_{(q, m, \delta_1)}$ has papameters $[n,\, m+1,\, \delta_1]$, and 
meets the Griesmer bound.   
\end{theorem}

We will need the following lemma shortly. 

\begin{lemma}\label{lem-2ndcosetleader}  
The second largest $q$-cyclotomic coset leader modulo $n$ is $\delta_2=(q-1)q^{m-1}-1-q^{\lfloor (m-1)/2\rfloor}$. Furthermore, 
\begin{eqnarray*}
\left|C_{\delta_2}\right|= \left\{ 
\begin{array}{ll}
m & \mbox{ if $m$ is odd,} \\
\frac{m}{2} & \mbox{ if $m$ is even.} 
\end{array}
\right. 
\end{eqnarray*}
\end{lemma} 

\begin{proof}
The proof is divided into the following two cases according to the parity of $m$. 

\subsubsection*{Case I, i.e., $m$ is odd}

In this case, we have 
$$ 
\delta_2=(q-1)q^{m-1}-1-q^{(m-1)/2}=n-\left(q^{(m-1)/2}+1\right)q^{(m-1)/2}. 
$$
It is easily seen that 
$$ 
q \delta_2 \bmod{n} = n-(q^{(m+1)/2}+1). 
$$
One can then verify that 
\begin{eqnarray*}
C_{\delta_2} &=& \left\{n-\left(q^{(m+1)/2}+1\right)q^i: i=0, 1, \ldots, \frac{m-3}{2}\right\} \bigcup \\
                      &  & \left\{n-\left(q^{(m-1)/2}+1\right)q^i: i=0, 1, \ldots, \frac{m-1}{2}\right\}.  
\end{eqnarray*}
Therefore, $\delta_2$ is the smallest integer in $C_{\delta_2}$ and is thus the coset leader. Clearly, we 
have $|C_{\delta_2}|=m$. 

Let $t=(m-3)/2$, $1 \leq i \leq q^{(m-1)/2}-1$ and let $J_i=(q-1)q^{m-1}-1-i$. Notice that 
$$ 
q^{(m-1)/2}-1=(q-1)q^t+(q-1)q^{t-1}+ \ldots + (q-1)q+q-1. 
$$  
The $q$-adic expansion of $i$ must be of the form 
$$
i=i_{t}q^t + i_{t-1}q^{t-1} + \ldots + i_1q+i_0, 
$$
where each $i_j$ satisfies $0 \leq i_j \leq q-1$, but at least one of the $i_j$'s is nonzero. 
It then follows 
that the $q$-adic expansion of $J_i$ is given by 
\begin{eqnarray*}
J_i &=& (q-2)q^{m-1}+(q-1)q^{m-2}+(q-1)q^{m-3}+ \ldots + (q-1)q^{t+1} + \\
     & &   (q-1-i_t)q^t +(q-1-i_{t-1})q^{t-1} + \ldots + (q-1-i_1)q + q-1-i_0. 
\end{eqnarray*} 

\subsubsection*{Subcase I.1, i.e., $q=2$} 

In this subcase, we have 
\begin{eqnarray*}
J_i = 2^{m-2}+2^{m-3}+ \ldots + 2^{t+1} +  (1-i_t)2^t +(1-i_{t-1})2^{t-1} + \ldots + (1-i_1)2 + 1-i_0. 
\end{eqnarray*} 

If $i_0=1$, then $J_i/2$ and $J_i$ are in the same $2$-cyclotomic coset modulo $n$. 
Hence, $J_i$ cannot be a coset leader. 

We now assume that $i_0=0$. Since $i \neq 0$, one of the $i_\ell$'s must be nonzero. Let $\ell$ denote 
the largest one such that $i_\ell =1$. One can then verify that 
$$ 
J_i 2^{m-1-\ell} \bmod{n} < J_i. 
$$
Whence, $J_i$ cannot be a coset leader. 

\subsubsection*{Subcase I.2, i.e., $q>2$} 

If $i_\ell >1$ for some $\ell$ with $0 \leq \ell \leq t$, then $J_i q^{m-1-\ell} \bmod{n}<J_i$. In this case, 
$J_i$ cannot be a coset leader. 

We now assume that all $i_\ell \in \{0,1\}$. Since $i \geq 1$, at least one of the $i_\ell$'s must be 1. 
 Let $\ell$ denote 
the largest one such that $i_\ell =1$. One can then verify that 
$$ 
J_i q^{m-1-\ell} \bmod{n} < J_i. 
$$
Whence, $J_i$ cannot be a coset leader. 

Summarizing all the conclusions above, we conclude that $\delta_2$ is the second largest coset leader for the case 
that $m$ is odd. 

\subsubsection*{Case II, i.e.,  $m$ is even}

In this case, we have 
$$ 
\delta_2=(q-1)q^{m-1}-1-q^{(m-2)/2}=n-\left(q^{m/2}+1\right)q^{(m-2)/2}. 
$$
It is easily seen that 
\begin{eqnarray*}
C_{\delta_2} &=& \left\{n-\left(q^{m/2}+1\right)q^i: i=0, 1, \ldots, \frac{m-2}{2}\right\}.  
\end{eqnarray*}
Therefore, $\delta_2$ is the smallest integer in $C_{\delta_2}$ and is the coset leader. Obviously, 
$|C_{\delta_2}|=m/2$. 

When $m=2$, $\delta_2=\delta_1+1$, where $\delta_1$ was defined in Lemma \ref{lem-1stcosetleader}.  
There does not exist any coset leader between $\delta_1$ 
and $\delta_2$. Therefore, we now assume that $m \geq 4$. 

Let $t=(m-4)/2$, $1 \leq i \leq q^{(m-2)/2}-1$ and let $J_i=(q-1)q^{m-1}-1-i$. Notice that 
$$ 
q^{(m-2)/2}-1=(q-1)q^t+(q-1)q^{t-1}+ \ldots + (q-1)q+q-1. 
$$  
The $q$-adic expansion of $i$ must be of the form 
$$
i=i_{t}q^t + i_{t-1}q^{t-1} + \ldots + i_1q+i_0, 
$$
where each $i_j$ satisfies $0 \leq i_j \leq q-1$, but at least one of the $i_j$'s is nonzero. 
It then follows 
that the $q$-adic expansion of $J_i$ is given by 
\begin{eqnarray*}
J_i &=& (q-2)q^{m-1}+(q-1)q^{m-2}+(q-1)q^{m-3}+ \ldots + (q-1)q^{t+1} + \\
     & &   (q-1-i_t)q^t +(q-1-i_{t-1})q^{t-1} + \ldots + (q-1-i_1)q + q-1-i_0. 
\end{eqnarray*} 

\subsubsection*{Subcase II.1, i.e., $q=2$} 

In this subcase, we have 
\begin{eqnarray*}
J_i = 2^{m-2}+2^{m-3}+ \ldots + 2^{t+1} +  (1-i_t)2^t +(1-i_{t-1})2^{t-1} + \ldots + (1-i_1)2 + 1-i_0. 
\end{eqnarray*} 

If $i_0=1$, then $J_i/2 < J_i$. But $J_i/2$ and $J_i$ are in the same $2$-cyclotomic coset modulo $n$. 
Hence, $J_i$ cannot be a coset leader. 

We now assume that $i_0=0$. Since $i \neq 0$, one of the $i_\ell$'s must be nonzero. Let $\ell$ denote 
the largest one such that $i_\ell =1$. One can then verify that 
$$ 
J_i 2^{m-1-\ell} \bmod{n} < J_i. 
$$
Whence, $J_i$ cannot be a coset leader. 

\subsubsection*{Subcase II.2, i.e., $q>2$} 

If $i_\ell >1$ for some $\ell$ with $0 \leq \ell \leq t$, then $J_i q^{m-1-\ell} \bmod{n}<J_i$. In this case, 
$J_i$ cannot be a coset leader. 

We now assume that all $i_\ell \in \{0,1\}$. Since $i \geq 1$, at least one of the $i_\ell$'s must be 1. 
 Let $\ell$ denote 
the largest one such that $i_\ell =1$. One can then verify that 
$$ 
J_i q^{m-1-\ell} \bmod{n} < J_i. 
$$
Whence, $J_i$ cannot be a coset leader. 

Summarizing all the conclusions above, we deduce that $\delta_2$ is the second largest coset leader for the case 
that $m$ is even. 
\end{proof}

\begin{table}[ht]
\caption{Weight distribution of $\tilde{\C}_{(2,\, m,\, \delta_2)}$ for odd $m$.}\label{tab-CG1}
\centering
\begin{tabular}{ll}
\hline
Weight $w$    & No. of codewords $A_w$  \\ \hline
$0$                                                        & $1$ \\
$2^{m-1}-2^{(m-1)/2}$           & $(2^m-1)(2^{(m-1)/2}+1)2^{(m-3)/2}$ \\
$2^{m-1}$                             & $(2^m-1)(2^{m-1}+1)$ \\
$2^{m-1}+2^{(m-1)/2}$           & $(2^m-1)(2^{(m-1)/2}-1)2^{(m-3)/2}$ \\ \hline
\end{tabular}
\end{table}

\begin{table}[ht]
\caption{Weight distribution of $\tilde{\C}_{(2,\, m,\, \delta_2)}$ for even $m$.}\label{tab-CG2}
\centering
\begin{tabular}{ll}
\hline
Weight $w$    & No. of codewords $A_w$  \\ \hline
$0$                                                        & $1$ \\
$2^{m-1}-2^{(m-2)/2}$          & $(2^{m/2}-1)(2^{m-1}+2^{(m-2)/2})$ \\
$2^{m-1}$                               & $2^m-1$ \\
$2^{m-1}+2^{(m-2)/2}$          & $(2^{m/2}-1)(2^{m-1}-2^{(m-2)/2})$ \\ \hline
\end{tabular}
\end{table}

\begin{table}[ht]
\caption{Weight distribution of $\tilde{\C}_{(q,\, m,\, \delta_2)}$ for odd $m$.}\label{tab-CG3}
\centering
\begin{tabular}{ll}
\hline
Weight $w$    & No. of codewords $A_w$  \\ \hline
$0$                                                        & $1$ \\ 
$(q-1)q^{m-1}-q^{(m-1)/2}$           & $(q-1)(q^m-1)(q^{m-1}+q^{(m-1)/2})/2$ \\ 
$(q-1)q^{m-1}$                            & $(q^m-1)(q^{m-1}+1)$ \\ 
$(q-1)q^{m-1}+q^{(m-1)/2}$          & $(q-1)(q^m-1)(q^{m-1}-q^{(m-1)/2})/2$ \\ \hline
\end{tabular}
\end{table}

\begin{table}[ht]
\caption{Weight distribution of $\tilde{\C}_{(q,\, m,\, \delta_2)}$ for even $m$.}\label{tab-CG22}
\centering
\begin{tabular}{ll}
\hline
Weight $w$    & No. of codewords $A_w$  \\ \hline
$0$                                                        & $1$ \\ 
$(q-1)q^{m-1}-q^{(m-2)/2}$          & $(q-1)(q^{(3m-2)/2}-q^{(m-2)/2})$ \\ 
$(q-1)q^{m-1}$                               & $q^m-1$ \\ 
$(q-1)(q^{m-1}+q^{(m-2)/2})$          & $q^{(m-2)/2}(q^m-q^{(m+2)/2}+q-1)$ \\ \hline
\end{tabular}
\end{table}

\begin{theorem}\label{thm-bchccc1}
The code $\tilde{\C}_{(q, m, \delta_2)}$ has parameters $[n,\, \tilde{k},\, \tilde{d}]$, where $\tilde{d} \geq \delta_2+1$ 
and 
\begin{eqnarray}
\tilde{k}=\left\{ \begin{array}{ll}
2m  & \mbox{ for odd } m, \\
\frac{3m}{2} & \mbox{ for even } m.  
\end{array}
\right.
\end{eqnarray} 

When $q=2$ and $m$ is odd, $\tilde{d} = \delta_2+1$ and the weight distribution of the code is given in Table  \ref{tab-CG1}. 
When $q=2$ and $m$ is even, $\tilde{d} = \delta_2+1$ and the weight distribution of the code is given in Table \ref{tab-CG2}.

When $q$ is an odd prime, $\tilde{d} = \delta_2+1$ and $\tilde{\C}_{(q,\, m,\, \delta_2)}$ is a three-weight code with the weight distribution of Table \ref{tab-CG3} for odd $m$ and Table \ref{tab-CG22} for even $m$.  
\end{theorem}

\begin{proof}
The conclusions on the dimension $\tilde{k}$ follow from Lemmas \ref{lem-1stcosetleader} and \ref{lem-2ndcosetleader}. 
By the BCH bound,  the minimum distance $\tilde{d} \geq \delta_2+1$. 

When $q=2$, $\tilde{d} = \delta_2+1$ and the weight distribution of $\tilde{\C}_{(q,\, m,\, \delta_2)}$ was determined in 
 \cite{Goet79,Gold,Kasa69}. 
 
We now treat the weight distribution of $\tilde{\C}_{(q,\, m,\, \delta_2)}$ for the case that $q$ is an odd prime. 
When $q$ is an odd prime and $m$ is odd, the weight distribution of $\tilde{\C}_{(q,\, m,\, \delta_2)}$ was settled 
in \cite{YCD06}. When $q$ is an odd prime and $m$ is even, the weight distribution of $\tilde{\C}_{(q,\, m,\, \delta_2)}$ 
is not documented in the literature, although the work in \cite{LK92} is related to and may be extended to settle 
the weight distribution of the code in this case. Hence, we will provide a proof of the weight distribution of 
$\tilde{\C}_{(q,\, m,\, \delta_2)}$ for both cases.  As will be seen below, our proof here can characterize all 
codewords of $\tilde{\C}_{(q,\, m,\, \delta_2)}$ with minimum weight. 
 
From now on in the proof, we assume that $q$ is odd prime.   Let $\eta'$ and $\eta$ denote the quadratic 
characters of $\gf(q^m)$ and $\gf(q)$, respectively. Let $\chi_1'$ and $\chi_1$ denote the canonical 
additive characters of $\gf(q^m)$ and $\gf(q)$, respectively. We will need the following results regarding 
Gauss sums \cite[Section 5.2]{LN}:  
\begin{eqnarray}\label{eqn-G1}
G(\eta, \chi_1)=\sum_{y \in \gf(q)^*} \eta(y) \chi_1(y)= \left\{ 
\begin{array}{ll}
\sqrt{q}   & \mbox{ if } q \equiv 1 \pmod{4} \\
\iota \sqrt{q}   & \mbox{ if } q \equiv 3 \pmod{4}
\end{array}
\right. 
\end{eqnarray}
where $\iota=\sqrt{-1}$ and  
\begin{eqnarray}\label{eqn-G2}
G(\eta, \chi_a)=\eta(a)G(\eta, \chi_1)
\end{eqnarray}
for all $a \in \gf(q)^*$, where $\chi_a(x)=\chi_1(ax)$ for all $x \in \gf(q)$.

It follows from the definition of $\tilde{\C}_{(q,\, m,\, \delta_2)}$ and Lemmas \ref{lem-1stcosetleader} and \ref{lem-2ndcosetleader} that the check polynomial of this code is $m_{\delta_1(x)}m_{\delta_2(x)}$. 
Notice that $\delta_1=n-q^{m-1}$ and 
$$ 
\delta_2=(q-1)q^{m-1}-1-q^{\lfloor (m-1)/2\rfloor} = n-(q^{m-1}+q^{\lfloor (m-1)/2\rfloor} ). 
$$

From Delsarte's Theorem \cite{Delsarte},  we then deduce that $\tilde{\C}_{(q,\, m,\, \delta_2)}$ is equivalent 
to  the following code (up to coordinate permutation) 
\begin{eqnarray}\label{eqn-codeCdelta2}
\tilde{\C}_{\delta_2}=\left\{\left(\tr\left(ax^{1+q^{\lfloor (m-1)/2\rfloor+1}} +bx\right)\right)_{x \in \gf(q^m)^*}: 
                           a \in \gf(q^m),\, b \in \gf(q^m)\right\},  
\end{eqnarray} 
herein and hereafter $\tr$ denotes the trace function from $\gf(q^m)$ to $\gf(q)$. In the definition of the code $\tilde{\C}_{\delta_2}$, 
we do not specify the order in which the elements of $\gf(q^m)$ are arranged when the codewords are defined, due to 
the fact that the codes resulted from different orderings of the elements of $\gf(q^m)^*$ are equivalent, and thus have 
the same weight distribution.

Define $h=\lfloor (m-1)/2\rfloor+1$ and let 
$$ 
f(x)=\tr\left(ax^{1+q^{h}} +bx\right) 
$$
where $a \in \gf(q^m)$ and $b \in \gf(q^m)$. 

We now consider the Hamming weight of the codeword 
$$ 
\bc_{(a,b)}=\left( f(x) \right)_{x \in \gf(q^m)^*}
$$
where $a \in \gf(q^m)$ and $b \in \gf(q^m)$.  It is straightforward to deduce that 
\begin{eqnarray}\label{eqn-wtform}
\wt(\bc_{(a,b)}) 
&=& (q-1)q^{m-1} - \frac{1}{q} \sum_{z \in \gf(q)^*} \sum_{x \in \gf(q^m)} \chi_1' (zf(x)) \nonumber \\
&=& (q-1)q^{m-1} - \frac{1}{q} \sum_{z \in \gf(q)^*} \sum_{x \in \gf(q^m)} \chi_1' (zax^{1+q^h} + zbx). 
\end{eqnarray}

We treat the weight distribution of $\tilde{\C}_{\delta_2}$ according to the parity of $m$ as follows. 

\subsection*{Case 1: $q$ is an odd prime and $m \geq 3$ is odd} 

In this case, we have the following basic facts that will be employed later: 
\begin{itemize}
\item[F1:] $h=(m+1)/2$. 
\item[F2:] $\gcd(h, m)=1$. 
\item[F3:] $\eta'(z)=\eta(z)$ for all $z \in \gf(q)^*$ (due to the fact that $(q^m-1)/(q-1)$ is odd). 
\item[F4:] $\chi_1'(x)=\chi_1(\tr(x))$ for all $x \in \gf(q^m)$. 
\item[F5:] $F(x):=a^{q^{h}}x^{q^{2h}}+ax=a^{q^{h}}x^{q}+ax$ is a permutation polynomial on $\gf(q^m)$ 
for each $a \in \gf(q^m)^*$, as $x^{1+q^h}$ is a planar monomial over $\gf(q^m)$. 
\end{itemize} 

\subsubsection*{Case 1.1: Let $a \neq 0$ and $b \neq 0$} 

Recall that $F(x)=a^{q^{h}}x^{q}+ax$ is a permutation polynomial over $\gf(q^m)$ for any $a \in \gf(q^m)^*$.  
Let $x_0$ be the unique solution of $F(x)=a^{q^{h}}x^{q}+ax=-b^{q^h}$ for any $a \in \gf(q^m)^*$ and 
$b \in \gf(q^m)$.  Put 
$$ 
u=\tr\left(ax_0^{1+q^h}\right). 
$$

In this subcase, it follows from Theorem 1 in \cite{Coulter} that 
\begin{eqnarray*}
\sum_{x \in \gf(q^m)} \chi_1' (af(x)) 
&=& \left\{ 
\begin{array}{ll}
q^{m/2} \overline{\eta'(-za) \chi_1'(za x_0^{1+q^h})},  & \mbox{ if } q \equiv 1 \pmod{4} \\
\iota^{3m} q^{m/2} \overline{\eta'(-za) \chi_1'(za x_0^{1+q^h})},  & \mbox{ if } q \equiv 3 \pmod{4}
\end{array}
\right. 
\\ 
&=& \left\{ 
\begin{array}{ll}
q^{m/2}\eta'(-a) \overline{\eta(z) \chi_1(zu)},  & \mbox{ if } q \equiv 1 \pmod{4}, \\
\iota^{3m} \eta'(-a) q^{m/2} \overline{\eta(z) \chi_1(zu)},  & \mbox{ if } q \equiv 3 \pmod{4}, 
\end{array}
\right. 
\end{eqnarray*}
where $\overline{\ell}$ denotes the complex conjugate of the complex number $\ell$.   

When $u=0$, we have 
\begin{eqnarray*}
\sum_{z \in \gf(q)^*} \sum_{x \in \gf(q^m)} \chi_1' (af(x)) 
&=& \left\{ 
\begin{array}{ll}
\sum_{z \in \gf(q)^*} q^{m/2}\eta'(-a) \eta(z),  & \mbox{ if } q \equiv 1 \pmod{4} \\
\sum_{z \in \gf(q)^*} \iota^{3m} \eta'(-a) q^{m/2} \eta(z),  & \mbox{ if } q \equiv 3 \pmod{4} 
\end{array}
\right. 
\\ 
&=& 0. 
\end{eqnarray*}
Consequently, $\wt(\bc_{(a,b)})=(q-1)q^{m-1}$. 

When $u \neq 0$, we have 
\begin{eqnarray*}
\sum_{z \in \gf(q)^*} \sum_{x \in \gf(q^m)} \chi_1' (af(x)) 
&=& \left\{ 
\begin{array}{ll}
 q^{m/2}\eta'(-a) \overline{G(\eta, \chi_u)},  & \mbox{ if } q \equiv 1 \pmod{4} \\
\iota^{3m} \eta'(-a) q^{m/2}  \overline{G(\eta, \chi_u)},  & \mbox{ if } q \equiv 3 \pmod{4} 
\end{array}
\right. 
\\ 
&=& \left\{ 
\begin{array}{ll}
 q^{(m+1)/2}\eta'(-a) \eta(u),  & \mbox{ if } q \equiv 1 \pmod{4}, \\
 q^{(m+1)/2}\iota^{3m+1} \eta'(-a) \eta(u),  & \mbox{ if } q \equiv 3 \pmod{4}.  
\end{array}
\right. 
\end{eqnarray*}
Consequently, $\wt(\bc_{(a,b)})=(q-1)q^{m-1} \pm q^{(m-1)/2}$. 

\subsubsection*{Case 1.2: Let $a \neq 0$ and $b =0$} 

In this case, it follows from Theorem 2.3 in \cite{Coulter} that 
\begin{eqnarray*}
\sum_{x \in \gf(q^m)} \chi_1' (af(x)) 
&=& \sum_{x \in \gf(q^m)} \chi_1' (azx^{1+q^h}) \\  
&=& \left\{ 
\begin{array}{ll}
q^{m/2} \eta'(za),  & \mbox{ if } q \equiv 1 \pmod{4} \\
\iota^{m} q^{m/2} \eta'(za),  & \mbox{ if } q \equiv 3 \pmod{4}
\end{array}
\right. 
\\ 
&=& \left\{ 
\begin{array}{ll}
q^{m/2} \eta'(a) \eta(z),  & \mbox{ if } q \equiv 1 \pmod{4}, \\
\iota^{m} q^{m/2} \eta'(a) \eta(z),  & \mbox{ if } q \equiv 3 \pmod{4}. 
\end{array}
\right. 
\end{eqnarray*}
As a result, we obtain 
\begin{eqnarray*}
\sum_{z \in \gf(q)^*} \sum_{x \in \gf(q^m)} \chi_1' (af(x)) =0. 
\end{eqnarray*} 
Hence, $\wt(\bc_{(a,b)})=(q-1)q^{m-1}$. 

\subsubsection*{Case 1.3: Let $a=0$ and $b \neq 0$} 

In this case, $f(x)=\tr(bx)$. Obviously, $\wt(\bc_{(0,b)})=(q-1)q^{m-1}$. 

Summarizing the conclusions of Cases 1.1, 1.2, and 1.3, we see that the code $\tilde{\C}_{\delta_2}$ has the 
following three nonzero weights: 
\begin{eqnarray*}
w_1 &=& (q-1)q^{m-1} - q^{(m-1)/2}, \\
w_2 &=& (q-1)q^{m-1}, \\
w_3 &=& (q-1)q^{m-1} + q^{(m-1)/2}. 
\end{eqnarray*}
Let $A_{w_i}$ be the total number of codewords with Hamming weight $w_i$ in $\tilde{\C}_{\delta_2}$. 
It is straightforward to see that the minimum distance of the dual of $\tilde{\C}_{\delta_2}$ is at least $3$. 
Then the first three Pless power moments yield the following set of equations \cite[p. 259]{HP03}: 
\begin{equation}
\left\{ 
\begin{array}{l}
A_{w_1}+A_{w_2}+A_{w_3}=q^{2m}-1, \\
w_1A_{w_1}+w_2A_{w_2}+w_3A_{w_3}=q^{2m-1}(q-1)(q^m-1), \\
w_1^2A_{w_1}+w_2^2A_{w_2}+w_3^2A_{w_3}=q^{2m-2}(q -1) (q^{m+1}-q^m-q+2). 
\end{array}
\right. 
\end{equation} 
Solving this system of equations gives the $A_{w_i}$'s in Table \ref{tab-CG3}. 

\subsection*{Case 2: $q$ is an odd prime and $m \geq 4$ is even} 

In this case, we have the following basic facts that will be used subsequently: 
\begin{itemize}
\item[H1:] $h=\lfloor (m-1)/2 \rfloor +1 =m/2$. 
\item[H2:] $\gcd(h, m)=h=m/2$. 
\item[H3:] $\eta'(z)=1$ for all $z \in \gf(q)^*$ (due to the fact that $(q^m-1)/(q-1)$ is even). 
\item[H4:] $\chi_1'(x)=\chi_1(\tr(x))$ for all $x \in \gf(q^m)$. 
\item[H5:] The equation $y^{q^h}+y=0$ has $q^h$ solutions $y$ in $\gf(q^m)$ (it follows from 
                  Lemma 2.2 in \cite{Coulter}). 
\item[H6:] $F(x):=a^{q^{h}}x^{q^{2h}}+ax=(a^{q^{h}}+a)x$ is a permutation polynomial on $\gf(q^m)$ 
for $q^m-q^{m/2}-1$ nonzero $a \in \gf(q^m)$, and is not a permutation polynomial on $\gf(q^m)$ for 
$q^{m/2}-1$ nonzero elements $a \in \gf(q^m)$. 
\end{itemize} 

We now consider the Hamming weight $\wt(\bc_{(a,b)})$ of the codeword $\bc_{(a,b)}$ case by case for Case 2. 

\subsubsection*{Case 2.1: Let $a \neq 0$,  $a^{q^h}+a \neq 0$ and $b \neq 0$} 

In this subcase, let $x_0=-b^{q^h}/(a^{q^h}+a)$. Then we have 
$$ 
ax_0^{1+q^h}=\frac{ab^{1+q^h}}{(a^{q^h}+a)^2}. 
$$
Put 
$$
u=\tr\left(ax_0^{1+q^h}\right)=\tr\left( \frac{ab^{1+q^h}}{(a^{q^h}+a)^2} \right). 
$$
It then follows from Theorem 1 in \cite{Coulter} that 
\begin{eqnarray*}
\sum_{x \in \gf(q^m)} \chi_1' (af(x)) 
= - q^{m/2} \overline{\chi_1' \left(azx_0^{1+q^h} \right)} 
= - q^{m/2} \overline{\chi_1 \left(zu \right)}.  
\end{eqnarray*}
We obtain then 
\begin{eqnarray*}
\sum_{z \in \gf(q)^*} \sum_{x \in \gf(q^m)} \chi_1' (af(x)) 
= \left\{ 
\begin{array}{ll}
-(q-1)q^{m/2},  & \mbox{ if } u =0, \\
q^{m/2},  & \mbox{ if } u  \neq 0. 
\end{array}
\right. 
\end{eqnarray*} 
Consequently, 
\begin{eqnarray*}
\wt(\bc_{(a,b)})
= \left\{ 
\begin{array}{ll}
(q-1)\left(q^{m-1}+q^{(m-2)/2}\right),  & \mbox{ if } u =0, \\
(q-1)q^{m-1}-q^{(m-2)/2},  & \mbox{ if } u  \neq 0. 
\end{array}
\right. 
\end{eqnarray*} 

\subsubsection*{Case 2.2: Let $a \neq 0$,  $a^{q^h}+a = 0$ and $b \neq 0$} 

In this subcase $F(x)=-b^{q^h}$ has no solution $x \in \gf(q^m)$. It then follows from Theorem 2 in \cite{Coulter} 
that 
$$ 
\sum_{x \in \gf(q^m)} \chi_1' (af(x)) =0. 
$$ 
As a result, we have 
$$
\sum_{z \in \gf(q)^* } \sum_{x \in \gf(q^m)} \chi_1' (af(x))=0.  
$$
It then follows from (\ref{eqn-wtform}) that $\wt(\bc_{(a,b)})=(q-1)q^{m-1}$. 

\subsubsection*{Case 2.3: Let $a \neq 0$ and $b = 0$}  

In this case, it follows from Theorem 2.4 in \cite{Coulter} that 
\begin{eqnarray*}
\sum_{x \in \gf(q^m)} \chi_1' (af(x)) =
\left\{ 
\begin{array}{ll}
-q^{m/2}  & \mbox{ if } a^{q^h}+a \neq 0, \\
q^{m}  & \mbox{ if } a^{q^h}+a = 0. 
\end{array}
\right. 
\end{eqnarray*} 
Hence, 
\begin{eqnarray*}
\sum_{z \in \gf(q)^* }  \sum_{x \in \gf(q^m)} \chi_1' (af(x))= 
\left\{ 
\begin{array}{ll}
-(q-1)q^{m/2}  & \mbox{ if } a^{q^h}+a \neq 0, \\
(q-1)q^{m}  & \mbox{ if } a^{q^h}+a = 0. 
\end{array}
\right. 
\end{eqnarray*}
We then deduce that 
\begin{eqnarray*}
\wt(\bc_{(a,b)})= 
\left\{ 
\begin{array}{ll}
(q-1)\left(q^{m-1} + q^{m/2} \right)  & \mbox{ if } a^{q^h}+a \neq 0, \\
0  & \mbox{ if } a^{q^h}+a = 0. 
\end{array}
\right. 
\end{eqnarray*}

\subsubsection*{Case 2.4: Let $a = 0$ and $b \neq 0$}  

In this subcase, we have $f(x)=\tr(bx)$. Obviously, $\wt(\bc_{(a,b)})= (q-1)q^{m-1}$. 

Summarizing the conclusions of Cases 2.1, 2.2, 2.3, and 2.4, we conclude that  $\tilde{\C}_{\delta_2}$ has the 
following three nonzero weights: 
\begin{eqnarray*}
w_1 &=& (q-1)q^{m-1} - q^{(m-2)/2}, \\
w_2 &=& (q-1)q^{m-1}, \\
w_3 &=& (q-1)\left(q^{m-1} + q^{(m-2)/2} \right). 
\end{eqnarray*}
Let $A_{w_i}$ be the total number of codewords with Hamming weight $w_i$ in $\tilde{\C}_{\delta_2}$. 
It is straightforward to see that the minimum distance of the dual of $\tilde{\C}_{\delta_2}$ is at least $3$. 
Then the first three Pless power moments yield the following set of equations: 
\begin{equation}
\left\{ 
\begin{array}{l}
A_{w_1}+A_{w_2}+A_{w_3}=q^{3m/2}-1, \\
w_1A_{w_1}+w_2A_{w_2}+w_3A_{w_3}=q^{(3m-2)/2}(q-1)(q^m-1), \\
w_1^2A_{w_1}+w_2^2A_{w_2}+w_3^2A_{w_3}=q^{(3m-4)/2}(q -1)(q^m-1) (q^{m+1}-q^m-q+2). 
\end{array}
\right. 
\end{equation} 
Solving this system of equations gives the $A_{w_i}$'s in Table \ref{tab-CG22}.  This completes the proof of 
this theorem. 
\end{proof} 

We remark that the proof of Theorem \ref{thm-bchccc1} given above may be modified to one for the same 
conclusions that hold for the case that $q$ is an odd prime power. To this end, the results in \cite{Coulter} 
should be generalized first. 

On the other hand, the proof of Theorem \ref{thm-bchccc1} actually characterizes all 
the codewords in $\tilde{\C}_{\delta_2}$ with the minimum Hamming weight. This characterisation is given 
in the next theorem and may be employed to derive $t$-designs with the code $\tilde{\C}_{\delta_2}$ (see 
\cite{AK92} for information about  $t$-designs associated with linear codes). 

\begin{theorem}\label{thm-minimumwt2}
Let $q$ be an odd prime, and let $m \geq 4$. When $m$ is odd, all the codewords of $\tilde{\C}_{\delta_2}$ 
with minimum weight $\tilde{d}=\delta_2+1$ are those $\bc_{(a,b)}$ such that 
\begin{itemize}
\item $ab \neq 0$ and $\eta'(-a)\eta\left(\tr\left(a x_0^{1+q^{(m+1)/2}}\right) \right)=1$ if $q \equiv 1 \pmod{4}$, where $x_0$ is 
          the unique solution of $a^{q^{(m+1)/2}}x^{q}+ax+b^{q^{(m+1)/2}}=0$; and 
\item $ab \neq 0$ and $\iota^{3m+1}\eta'(-a)\eta\left(\tr\left(a x_0^{1+q^{(m+1)/2}}\right) \right)=1$ if $q \equiv 3 \pmod{4}$, where $x_0$ is 
          the unique solution of $a^{q^{(m+1)/2}}x^{q}+ax+b^{q^{(m+1)/2}}=0$.            
\end{itemize}
When $m$ is even, all the codewords of $\tilde{\C}_{\delta_2}$ 
with minimum weight $\tilde{d}=\delta_2+1$ are those $\bc_{(a,b)}$ such that $a^{q^{m/2}}+a \neq 0$, $b \neq 0$ 
and 
$$ 
\tr\left( \frac{ab^{1+q^{m/2}}}{(a^{q^{m/2}}+a)^2} \right) \neq 0. 
$$
\end{theorem}

\begin{table}[ht]
\begin{center} 
\caption{Examples of  $\tilde{\C}_{(q,\, m,\, \delta_2)}$ of Theorem \ref{thm-bchccc1} }\label{tab-bchthm1}
\begin{tabular}{rrrrrl} \hline
$n$ &   $k$   & $d=\delta_2+1$  & $m$ & $q$  & Optimality  \\ \hline  
 $15$ & $6$   & $6$                           & $4$  & $1$ & Optimal                    \\ 
$31$ & $10$   & $12$                           & $5$  & $2$ &  Optimal                  \\ 
$63$ & $9$   & $28$                           & $6$  & $2$ &  Optimal                  \\ 
$127$ & $14$   & $56$                           & $7$  & $2$ & Optimal                     \\ 
$256$ & $12$   & $120$                           & $8$  & $2$ & Best known                    \\  
$26$ & $6$   & $15$                           & $3$  & $3$ &    Optimal                \\ 
$80$ & $6$   & $51$                           & $4$  & $3$ &  Optimal                   \\ 
$242$ & $10$   & $153$                           & $5$  & $3$ & Best known                    \\ \hline
\end{tabular}
\end{center} 
\end{table} 

Examples of the code $\tilde{\C}_{(q,\, m,\, \delta_2)}$ are summarized in Table \ref{tab-bchthm1}. They 
either are optimal,  or have the same parameters as the best linear codes in the 
Database.

The following theorem is proved in \cite{Gold,Kasa69}. 

\begin{theorem}\label{thm-dualdis1181}
The minimum distance $\tilde{d}^\perp$ of the dual of $\tilde{\C}_{(2, m, \delta_2)}$ is equal to $5$ when $m \geq 5$ 
is odd, and $3$ when $m \geq 4$ is even. 
\end{theorem}

\begin{theorem}\label{thm-bchccc2}
The code $\C_{(q,\, m,\, \delta_2)}$ has papameters $[n,\, k,\, d]$, where $d \geq \delta_2$ and 
\begin{eqnarray}
k=\left\{ \begin{array}{ll}
2m+1  & \mbox{ for odd } m, \\
\frac{3m}{2} +1 & \mbox{ for even } m.  
\end{array}
\right.
\end{eqnarray}
Furthermore, $d = \delta_2$ if $q$ is a prime.  
\end{theorem}

\begin{proof}
The conclusions on the dimension $k$ follow from Lemmas \ref{lem-1stcosetleader} and \ref{lem-2ndcosetleader}. 
By the BCH bound,  the minimum distance $d \geq \delta_2$. 

It follows from the definition of $\C_{(q,\, m,\, \delta_2)}$ and Lemmas \ref{lem-1stcosetleader} and \ref{lem-2ndcosetleader} that the check polynomial of this code is $(x-1)m_{\delta_1(x)}m_{\delta_2(x)}$. 
Notice that $\delta_1=n-q^{m-1}$ and 
$$ 
\delta_2=(q-1)q^{m-1}-1-q^{\lfloor (m-1)/2\rfloor} = n-(q^{m-1}+q^{\lfloor (m-1)/2\rfloor} ). 
$$

From Delsarte's Theorem we then deduce that $\C_{(q,\, m,\, \delta_2)}$ is equivalent to  
the following code 
\begin{eqnarray}\label{eqn-aug21}
\C_{\delta_2}=\left\{\left(\tr\left(ax+bx^{1+q^{\lfloor (m-1)/2\rfloor+1}}\right)+c\right)_{x \in \gf(q^m)^*}: 
                           a \in \gf(q^m),\, b \in \gf(q^m),\, c \in \gf(q) \right\}. 
\end{eqnarray}

To prove that $d=\delta_2$ for the case that $q$ is a prime, one can refine the proof of Theorem \ref{thm-bchccc1} 
with the quadratic expression of (\ref{eqn-aug21}) to obtain the weight distribution of the code. We leave the details 
to interested readers. 
 
\end{proof}

\begin{table}[ht]
\begin{center} 
\caption{Examples of  $\C_{(q,\, m,\, \delta_2)}$ of Theorem \ref{thm-bchccc2} }\label{tab-bchthm2}
\begin{tabular}{rrrrrl} \hline
$n$ &   $k$   & $d=\delta_2$  & $m$ & $q$  & Optimality  \\ \hline  
$15$ & $7$   & $5$                           & $4$  & $2$ &  Yes                 \\ 
$31$ & $11$   & $11$                           & $5$  & $2$ &  Yes                  \\ 
$63$ & $10$   & $27$                           & $6$  & $2$ & Best known                   \\ 
$127$ & $15$   & $55$                           & $7$  & $2$ & Best known                  \\ 
$255$ & $13$   & $119$                           & $8$  & $2$ &  Best known                  \\  
$26$ & $7$   & $14$                           & $3$  & $3$ &  Optimal                  \\ 
 $80$ & $7$   & $50$                           & $4$  & $3$ &  Optimal                  \\ 
$242$ & $11$   & $152$                           & $5$  & $3$ &  Best known                  \\ \hline
\end{tabular}
\end{center} 
\end{table} 

Examples of the code $\C_{(q,\, m,\, \delta_2)}$ are summarized in Table \ref{tab-bchthm2}. They 
are sometimes optimal, and sometimes have the same parameters as the best linear codes in the 
Database. When $(q,m)=(2,6)$, the cyclic code $\C_{(q,\, m,\, \delta_2)}$ has parameters 
$[63,10,27]$, which are the best possible parameters according to \cite[p. 258]{Dingbk15}.   
 When $(q,m)=(3,3)$, the cyclic code $\C_{(q,\, m,\, \delta_2)}$ has parameters 
$[26,7,14]$, which are the best possible parameters according to \cite[p. 300]{Dingbk15}.

\section{The parameters of the BCH codes $\C_{(q,\, m,\, \delta_3)}$ and $\tilde{\C}_{(q,\, m,\, \delta_3)}$, where $\delta_3=(q-1)q^{m-1}-1-q^{\lfloor (m+1)/2\rfloor}$} 

Before determining the parameters of the codes $\C_{(q,\, m,\, \delta_3)}$ and $\tilde{\C}_{(q,\, m,\, \delta_3)}$, 
we have to prove the following lemma. 

\begin{lemma}\label{lem-3rdcosetleader}  
Let $m \geq 4$. 
Then the third largest $q$-cyclotomic coset leader modulo $n$ is $\delta_3=(q-1)q^{m-1}-1-q^{\lfloor (m+1)/2\rfloor}$. In addition, 
$|\delta_3|=m$. 
\end{lemma}

\begin{proof}
The proof is divided into the following two cases according to the parity of $m$. 

\subsubsection*{Case I, i.e., $m$ is odd}

In this case, we have 
$$ 
\delta_3=(q-1)q^{m-1}-1-q^{(m+1)/2}=n-\left(q^{(m-3)/2}+1\right)q^{(m+1)/2}. 
$$
It can be verified that 
\begin{eqnarray*}
C_{\delta_3} &=& \left\{n-\left(q^{(m+3)/2}+1\right)q^i: i=0, 1, \ldots, \frac{m-5}{2}\right\} \bigcup \\
                      &  & \left\{n-\left(q^{(m-3)/2}+1\right)q^i: i=0, 1, \ldots, \frac{m+1}{2}\right\}.  
\end{eqnarray*}
Therefore, $\delta_3$ is the smallest integer in $C_{\delta_3}$ and is thus the coset leader. Clearly, we 
have $|C_{\delta_3}|=m$. 

Let $t=(m-1)/2$. By definition, 
\begin{eqnarray*}
\delta_2 &=& (q-1)q^{m-1}-1-q^{(m-1)/2} \\
&=& (q-2)q^{m-1} +  (q-1)q^{m-2} +   (q-1)q^{m-3} + \ldots +  (q-1)q^{t+1} + \\
& &  (q-2)q^{t} +  (q-1)q^{t-1} +   (q-1)q^{t-2} + \ldots +  (q-1)q + (q-1).    
\end{eqnarray*}
Observe that $\delta_2-\delta_3=(q-1)q^{t}$. We need to prove that $J_i:=\delta_2-i$ is not a coset leader 
for all $i$ with $1 \leq i \leq (q-1)q^{t}-1$. 

Notice that 
$$ 
(q-1)q^{t}-1=(q-2)q^t+(q-1)q^{t-1}+(q-1)q^{t-2}+ \ldots + (q-1)q+q-1. 
$$  
The $q$-adic expansion of $i$ must be of the form 
$$
i=i_{t}q^t + i_{t-1}q^{t-1} + \ldots + i_1q+i_0, 
$$
where  $i_\ell$ satisfies $0 \leq i_\ell \leq q-1$ for all $0 \leq \ell \leq t-1$ and $0 \leq i_t \leq q-2$, but at least one 
of the $i_\ell$'s is nonzero. 
It then follows 
that the $q$-adic expansion of $J_i$ is given by 
\begin{eqnarray*}
J_i &=& (q-2)q^{m-1}+(q-1)q^{m-2}+(q-1)q^{m-3}+ \ldots + (q-1)q^{t+1} + \\
     & &   (q-2-i_t)q^t +(q-1-i_{t-1})q^{t-1} + (q-1-i_{t-2})q^{t-2} +  \ldots + (q-1-i_1)q + q-1-i_0. 
\end{eqnarray*} 

\subsubsection*{Subcase I.1, i.e., $q=2$} 

In this subcase, we have $i_t=0$ and 
\begin{eqnarray*}
J_i = 2^{m-2}+2^{m-3}+ \ldots + 2^{t+1}  +(1-i_{t-1})2^{t-1} + \ldots + (1-i_1)2 + 1-i_0. 
\end{eqnarray*} 

If $i_0=1$, then $J_i/2 < J_i$. But $J_i/2$ and $J_i$ are in the same $2$-cyclotomic coset modulo $n$. 
Hence, $J_i$ cannot be a coset leader. 

We now assume that $i_0=0$. Since $i \neq 0$, one of the $i_\ell$'s must be nonzero. Let $\ell$ denote 
the largest one such that $i_\ell =1$. One can then verify that 
$$ 
J_i 2^{m-1-\ell} \bmod{n} < J_i. 
$$
Whence, $J_i$ cannot be a coset leader. 

\subsubsection*{Subcase I.2, i.e., $q>2$} 

If $i_t \geq 1$, then $J_i q^{m-1-t} \bmod{n}<J_i$. In this case, 
$J_i$ cannot be a coset leader. 

If $i_\ell \geq 2$ for some $\ell$ with $0 \leq \ell \leq t-1$, then $J_i q^{m-1-\ell} \bmod{n}<J_i$. In this case, 
$J_i$ cannot be a coset leader. 

We now assume that all $i_\ell \in \{0,1\}$ for all $0 \leq \ell \leq t-1$ and $i_t=0$. Since $i \geq 1$, at least one of the $i_\ell$'s must be 1. 
 Let $\ell$ denote 
the largest one such that $i_\ell =1$. One can then verify that 
$$ 
J_i q^{m-1-\ell} \bmod{n} < J_i. 
$$
Whence, $J_i$ cannot be a coset leader. 

Summarizing all the conclusions above, we deduce that $\delta_3$ is the third largest coset leader for the case 
that $m$ is odd. 

\subsubsection*{Case II, i.e.,  $m$ is even}

In this case, we have 
$$ 
\delta_3=(q-1)q^{m-1}-1-q^{(m+2)/2}=n-\left(q^{(m-4)/2}+1\right)q^{(m+2)/2}. 
$$
It is easily seen that 
\begin{eqnarray*}
C_{\delta_3} &=& \left\{n-\left(q^{(m-4)/2}+1\right)q^i: i=0, 1, \ldots, \frac{m+2}{2}\right\} \bigcup \\
                      & & \left\{n-\left(q^{(m+4)/2}+1\right)q^i: i=0, 1, \ldots, \frac{m-6}{2}\right\}. 
\end{eqnarray*}
Therefore, $\delta_3$ is the smallest integer in $C_{\delta_3}$ and is the coset leader. Obviously, 
$|C_{\delta_3}|=m$. 

Similarly as in the case that $m$ is odd, one can prove that $\delta_3$ is the third largest coset leader for the case 
that $m$ is even. Details are omitted here. 
\end{proof}

\begin{table}[ht]
\caption{The weight distribution of $\tilde{\C}_{(2,\, m,\, \delta_3)}$ for odd $m$.}\label{tab-zhou3}
\centering
\begin{tabular}{ll}
\hline
Weight $w$    & No. of codewords $A_w$  \\ \hline
$0$                                                        & $1$ \\ 
$2^{m-1}-2^{(m+1)/2}$           & $(2^m-1)\cdot 2^{(m-5)/2}\cdot (2^{(m-3)/2}+1)\cdot (2^{m-1}-1)/3$ \\ 
$2^{m-1}-2^{(m-1)/2}$           & $(2^m-1)\cdot 2^{(m-3)/2}\cdot (2^{(m-1)/2}+1)\cdot (5\cdot 2^{m-1}+4)/3$ \\ 
$2^{m-1}$           &           ${(2^m-1)}\cdot (9\cdot 2^{2m-4}+3\cdot 2^{m-3}+1)$ \\ 
$2^{m-1}+2^{(m-1)/2}$           & $(2^m-1)\cdot 2^{(m-3)/2}\cdot (2^{(m-1)/2}-1)\cdot (5\cdot 2^{m-1}+4)/3$ \\ 
$2^{m-1}+2^{(m+1)/2}$           & $(2^m-1)\cdot 2^{(m-5)/2}\cdot (2^{(m-3)/2}-1)\cdot (2^{m-1}-1)/3$ \\ \hline
\end{tabular}
\end{table}

\begin{table}[ht]
\caption{The weight distribution of $\tilde{\C}_{(2,\, m,\, \delta_3)}$ for even $m$.}\label{tab-zengli3}
\centering
\begin{tabular}{ll}
\hline
Weight $w$    & No. of codewords $A_w$  \\ \hline
$0$                                                        & $1$ \\ 
$2^{m-1}-2^{m/2}$           &  $(2^{m/2}-1) (2^{m-3} + 2^{(m-4)/2})(2^{m+1}+2^{m/2}-1)/3$ \\ 
$2^{m-1}-2^{(m-2)/2}$           &  $(2^{m/2}-1) (2^{m-1} + 2^{(m-2)/2})(2^m+2^{(m+2)/2}+4)/3$ \\ 
$2^{m-1}$           &   $(2^{m/2}-1)(2^{2m-1} + 2^{(3m-4)/2}-2^{m-2}+2^{m/2} +1)$         \\ 
$2^{m-1}+2^{(m-2)/2}$           &  $(2^{m/2}-1) (2^{m-1} - 2^{(m-2)/2})(2^m+2^{(m+2)/2}+4)/3$\\ 
$2^{m-1}+2^{m/2}$           & $(2^{m/2}-1) (2^{m-3} - 2^{(m-4)/2})(2^{m+1}+2^{m/2}-1)/3$ \\ \hline
\end{tabular}
\end{table}

\begin{theorem}\label{thm-bchccc3}
Let $m \geq 4$. 
The code $\tilde{\C}_{(q, m, \delta_3)}$ has papameters $[n,\, \tilde{k},\, \tilde{d}]$, where $\tilde{d} \geq \delta_3+1=(q-1)q^{m-1}-q^{\lfloor (m+1)/2\rfloor}$ and 
\begin{eqnarray}
\tilde{k}=\left\{ \begin{array}{ll}
3m  & \mbox{ for odd } m, \\
\frac{5m}{2} & \mbox{ for even } m.  
\end{array}
\right.
\end{eqnarray}

Furthermore, when $q=2$ and $m$ is odd, the binary code $\tilde{\C}_{(q, m, \delta_3)}$ has minimum distance 
$\tilde{d} = \delta_3+1$ and its weight distribution is given in Table  \ref{tab-zhou3}. 
When $q=2$ and $m$ is even, the binary code $\tilde{\C}_{(q, m, \delta_3)}$ has minimum distance 
$\tilde{d} = \delta_3+1$ and its weight distribution is given in Table \ref{tab-zengli3}. 

When $q$ is an odd prime and $m \geq 4$ is even, the code $\tilde{\C}_{(q, m, \delta_3)}$ has minimum distance 
$\tilde{d} = \delta_3+1$ and its weight distribution is given in Table  \ref{tab-zengli4}. 

When $q$ is an odd prime and $m \geq 5$ is odd, the code $\tilde{\C}_{(q, m, \delta_3)}$ has minimum distance 
$\tilde{d} = \delta_3+1$ and its weight distribution is given in Table  \ref{tab-zengli5}. 

\end{theorem}

\begin{proof}
The conclusions on the dimension $\tilde{k}$ follow from Lemmas \ref{lem-1stcosetleader}, \ref{lem-2ndcosetleader} 
and \ref{lem-3rdcosetleader}. 
By the BCH bound,  the minimum distance $\tilde{d} \geq \delta_3+1$. 

It follows from the definition of $\tilde{\C}_{(q,\, m,\, \delta_3)}$ and Lemmas \ref{lem-1stcosetleader},  \ref{lem-2ndcosetleader} and \ref{lem-3rdcosetleader} that the check polynomial of this code is 
$m_{\delta_1(x)}m_{\delta_2(x)}m_{\delta_3(x)}$. 
Notice that 
$$ 
\delta_3=(q-1)q^{m-1}-1-q^{\lfloor (m+1)/2\rfloor} = n-(q^{m-1}+q^{\lfloor (m+1)/2\rfloor} ). 
$$

From Delsarte's Theorem we then deduce that $\tilde{\C}_{(q,\, m,\, \delta_3)}$ is equivalent to 
the following code 
\begin{eqnarray*}
\tilde{\C}_{\delta_3}=\left\{
\begin{array}{r}
\left(\tr\left(ax+bx^{1+q^{h} }+cx^{1+q^{h+1}}\right)\right)_{x \in \gf(q^m)^*}:   
\\
        a \in \gf(q^m),\, b \in \gf(q^m),\,  c \in \gf(q^m)
\end{array}
        \right\}, 
\end{eqnarray*} 
where $h=\lfloor (m-1)/2\rfloor+1$. 

When $q=2$, the binary code $\tilde{\C}_{\delta_3}$ has minimum distance 
$\tilde{d} = \delta_3+1$ and its weight distribution was settled in  \cite{Kasa69}. 

When $q$ is an odd prime and $m \geq 4$ is even, the code $\tilde{\C}_{\delta_3}$ has minimum distance 
$\tilde{d} = \delta_3+1$ and its weight distribution in Table \ref{tab-zengli4} is a special case of Table  2 in 
\cite{zengseta}. 

When $q$ is an odd prime and $m \geq 5$ is odd, we have $h=(m+1)/2$ and $h+1=(m+3)/2$. It is easy to 
see that $\gcd(m, h)=1$ and  
$$ 
1+q^{3h} \equiv  1+q^{h+1} \pmod{n}. 
$$ 
It then follows that 
\begin{eqnarray*}
\tilde{\C}_{\delta_3}=\left\{
\begin{array}{r}
\left(\tr\left(ax+bx^{1+q^{h} }+cx^{1+q^{3h}}\right)\right)_{x \in \gf(q^m)^*}:   
\\
        a \in \gf(q^m),\, b \in \gf(q^m),\,  c \in \gf(q^m)
\end{array}
        \right\}.  
\end{eqnarray*} 
In this case, the weight distribution of $\tilde{\C}_{\delta_3}$ is a special case of Theorem  2 in 
\cite{Zeng10}. 

\end{proof}

\begin{table}[ht]
\caption{The weight distribution of $\tilde{\C}_{(q,\, m,\, \delta_3)}$ for even $m$ and odd $q$.}\label{tab-zengli4}
\centering
\begin{tabular}{ll}
\hline
Weight $w$    & No. of codewords $A_w$  \\ \hline
$0$                                                        & $1$ \\ 
$(q-1)q^{m-1}-q^{m/ 2}$   & $(q^m-1) ((q^2-1) (q^{(3 m-6) /2}+q^{m-2}) + 
        2 (q^{(m-2) /2}-1) (q^{m-3}+q^{(m-4) /2})) /2 (q+1)$ \\ 
$(q-1) (q^{m-1}-q^{(m-2) /2})$   & $q (q^{m /2}+1) (q^m-1) (q^{m-1}+(q-1) q^{(m-2) /2}) /2 (q+1)$ \\ 
$(q-1) q^{m-1}-q^{(m-2) /2}$   & $(q^{m+1}-2 q^m+q) (q^{m /2}-1) (q^{m-1}+q^{(m-2) /2}) /2$ \\ 
$(q-1) q^{m-1}$   & $(q^m-1) (1+q^{(3 m-2) /2}-q^{(3 m-4) /2} +2 q^{(3 m-6) /2}-q^{m-2})$ \\ 
$(q-1) q^{m-1}+q^{(m-2) /2}$   & $ q (q^{m /2}+1) (q^m-1) (q-1) (q^{m-1}-q^{(m-2) /2}) /2 (q+1)$ \\ 
$ (q-1) (q^{m-1}+q^{(m-2) /2})$   & $(q^{m+1}-2 q^m+q) (q^{m /2}-1) (q^{m-1}-(q-1) q^{(m-2) /2}) /2 (q-1) $ \\ 
$(q-1) q^{m-1}+q^{m /2}$   & $ q^{(m-2) /2} (q^m-1) (q-1) (q^{m-2}-q^{(m-2) /2}) /2 $ \\ 
$(q-1) (q^{m-1}+q^{m /2}) $   & $ (q^{(m-2) /2}-1) (q^m-1) (q^{m-3}-(q-1) q^{(m-4) /2}) /(q^2-1)$ \\ \hline  
\end{tabular}
\end{table}

\begin{table}[ht]
\caption{The weight distribution of $\tilde{\C}_{(q,\, m,\, \delta_3)}$ for odd $m$ and odd $q$.}\label{tab-zengli5}
\centering
\begin{tabular}{ll}
\hline
Weight $w$    & No. of codewords $A_w$  \\ \hline
$0$                                                        & $1$ \\ 
$(q-1)q^{m-1}-q^{(m+1)/2}$     & $(q^m-1)(q^{m-3}+q^{(m-3)/2})(q^{m-1}-1) / 2(q+1)$ \\
$(q-1)(q^{m-1}-q^{(m-1)/2})$     & $(q^m-1)(q^{m-1}+q^{(m-1)/2})(q^{m-2}+(q-1)q^{(m-3)/2})/2$ \\
$(q-1)q^{m-1}-q^{(m-1)/2}$     & $(q^m-1)(q^{m-2}+q^{(m-3)/2})(q^{m+3}-q^{m+2}-q^{m-1}-q^{(m+3)/2}+q^{(m-1)/2} +q^3 )/2(q+1)$ \\
$(q-1)q^{m-1}$     & $(q^m-1)(1+(q^2-q+1)q^{m-3}+(q-1)q^{2m-4}+(q-2)q^{2m-2} +q^{2m-1})$ \\
$(q-1)q^{m-1}+q^{(m-1)/2}$     & $(q^m-1)(q^{m-2}-q^{(m-3)/2})(q^{m+3}-q^{m+2}-q^{m-1}+q^{(m+3)/2}-q^{(m-1)/2} +q^3 )/2(q+1)$ \\
$(q-1)(q^{m-1}+q^{(m-1)/2})$     & $(q^m-1)(q^{m-1}-q^{(m-1)/2})(q^{m-2}-(q-1)q^{(m-3)/2})/2$ \\
$(q-1)q^{m-1}+q^{(m+1)/2}$     & $(q^m-1)(q^{m-3}-q^{(m-3)/2})(q^{m-1}-1) / 2(q+1)$ \\
 \hline  
\end{tabular}
\end{table}

The following theorem is proved in \cite{Kasa69}. 

\begin{theorem}\label{thm-dualdis181}
The minimum distance $\tilde{d}^\perp$ of the dual of $\tilde{\C}_{(2, m, \delta_3)}$ is equal to $7$ when $m \geq 5$ 
is odd, and $5$ when $m \geq 6$ is even. 
\end{theorem}

\begin{table}[ht]
\begin{center} 
\caption{Examples of  $\tilde{\C}_{(q,\, m,\, \delta_3)}$ of Theorem \ref{thm-bchccc3} }\label{tab-bchthm3}
\begin{tabular}{rrrrrl} \hline
$n$ &   $k$   & $d=\delta_3+1$  & $m$ & $q$  & Optimality  \\ \hline  
$15$ & $10$   & $4$                           & $4$  & $2$ &   Yes                 \\ 
$31$ & $15$   & $8$                           & $5$  & $2$ &   Yes                \\ 
$63$ & $15$   & $24$                           & $6$  & $2$ &    Yes                \\ 
$127$ & $21$   & $48$                           & $7$  & $2$ &  Best known                  \\ 
$255$ & $20$   & $112$                           & $8$  & $2$ &  Best known                  \\  
$26$ & $10$   & $9$                           & $3$  & $3$ &   No                 \\ 
 $80$ & $10$   & $45$                           & $4$  & $3$ &  Best known                  \\ 
$242$ & $15$   & $135$                           & $5$  & $3$ & No                   \\ \hline
\end{tabular}
\end{center} 
\end{table}

\begin{example} 
Let $(q, m)=(2,4)$. Then $\delta_3=3$, and $\tilde{\C}_{(q,\, m,\, \delta_3)}$ has parameters $[15, 10, 4]$ and 
weight enumerator $1+105z^4 + 280z^6 + 435z^8 + 168z^{10} + 35z^{12}$. 
\end{example} 

\begin{example} 
Let $(q, m)=(2,5)$. Then $\delta_3=7$, and $\tilde{\C}_{(q,\, m,\, \delta_3)}$ has parameters $[31, 15, 8]$ and 
weight enumerator $1+465z^8 + 8680z^{12} + 18259z^{16} + 5208z^{20} + 155z^{24}$. 
\end{example} 

\begin{example} 
Let $(q, m)=(3,4)$. Then $\delta_3=44$, and $\tilde{\C}_{(q,\, m,\, \delta_3)}$ has parameters $[80, 10, 45]$ and 
weight enumerator 
$$
1+3040z^{45} + 9900z^{48} + 10080z^{51} + 16640z^{54} + 14400z^{57} +3528z^{60} + 1440z^{63} +20z^{72}. 
$$
\end{example} 

\begin{example} 
Let $(q, m)=(3,5)$. Then $\delta_3=134$, and $\tilde{\C}_{(q,\, m,\, \delta_3)}$ has parameters $[242, 15, 135]$ and 
weight enumerator 
$$
1+29040z^{135} + 359370z^{144} +3855060 z^{153} + 6719372z^{162} + 3188592z^{171} + 182952z^{180} 
+ 14520z^{189}. 
$$
\end{example} 

The optimality of the code $\tilde{\C}_{(q,\, m,\, \delta_3)}$ is marked in Table \ref{tab-bchthm3}, where further examples 
of the code is documented. As shown in this table, the code $\tilde{\C}_{(q,\, m,\, \delta_3)}$ is sometimes optimal, and 
sometimes has the same parameters as the best linear code known.

\begin{theorem}\label{thm-bchccc4}
Let $m \geq 4$. 
The code $\C_{(q,\, m,\, \delta_3)}$ has papameters $[n,\, k,\, \delta_3]$, where where $d = \delta_3=(q-1)q^{m-1}-1-q^{\lfloor (m+1)/2\rfloor}$ and 
\begin{eqnarray}
k=\left\{ \begin{array}{ll}
3m+1  & \mbox{ for odd } m, \\
\frac{5m}{2}+1 & \mbox{ for even } m.  
\end{array}
\right.
\end{eqnarray}
\end{theorem}

\begin{proof}
The conclusions on the dimension $k$ follow from Lemmas \ref{lem-1stcosetleader}, \ref{lem-2ndcosetleader} and 
\ref{lem-3rdcosetleader}. 
By the BCH bound,  the minimum distance $d \geq \delta_3$. 

It follows from the definition of $\C_{(q,\, m,\, \delta_3)}$ and Lemmas \ref{lem-1stcosetleader}, 
\ref{lem-2ndcosetleader} and \ref{lem-3rdcosetleader} that the check polynomial of this code 
is $(x-1)m_{\delta_1(x)}m_{\delta_2(x)}m_{\delta_3(x)}$. 
Notice that 
$$ 
\delta_3=(q-1)q^{m-1}-1-q^{\lfloor (m+1)/2\rfloor} = n-(q^{m-1}+q^{\lfloor (m+1)/2\rfloor} ). 
$$

From Delsarte's Theorem we then deduce that $\C_{(q,\, m,\, \delta_3)}$ is equal to the following 
code 
\begin{eqnarray*}
\C_{\delta_3} =  
 \left\{
\begin{array}{r}
\left(\tr\left(ax+bx^{1+q^{\lfloor (m-1)/2\rfloor+1}}+cx^{1+q^{\lfloor (m+1)/2\rfloor+1}}\right)+e\right)_{x \in \gf(q^m)^*}:  \\
                           a \in \gf(q^m),\, b \in \gf(q^m),\,  c \in \gf(q^m),\, e \in \gf(q)
\end{array} 
\right\}. 
\end{eqnarray*}

Similarly, the weights and their frequencies of the codewords in $\C_{(q,\, m,\, \delta_3)}$ are determined by the 
affine and quadratic functions  
$$ 
\tr\left(ax+bx^{1+q^{\lfloor (m-1)/2\rfloor+1}}+cx^{1+q^{\lfloor (m+1)/2\rfloor+1}}\right)+e. 
$$
One can refine the proofs in \cite{Kasa69}, \cite{zengseta}, \cite{Zeng10} and \cite{ZhouDCC}, to prove that $d=\delta_3$. 
We omit the lengthy details here. 
\end{proof}

\begin{table}[ht]
\begin{center} 
\caption{Examples of  $\C_{(q,\, m,\, \delta_3)}$ of Theorem \ref{thm-bchccc4} }\label{tab-bchthm4}
\begin{tabular}{rrrrrl} \hline
$n$ &   $k$   & $d=\delta_3$  & $m$ & $q$  & Optimality  \\ \hline  
$15$ & $11$   & $3$                           & $4$  & $2$ &  Yes                  \\ \hline
$31$ & $16$   & $7$                           & $5$  & $2$ & No (optimal $d=8$)                   \\ \hline
$63$ & $16$   & $23$                           & $6$  & $2$ & Best known                   \\ \hline
$127$ & $22$   & $47$                           & $7$  & $2$ & Best known                   \\ \hline
$255$ & $21$   & $111$                           & $8$  & $2$ & Best known                   \\ \hline 
$26$ & $11$   & $8$                           & $3$  & $3$ &     No (best $d=9$)               \\ \hline
 $81$ & $11$   & $44$                           & $4$  & $3$ &   No (best $d=45$)                 \\ \hline
$242$ & $16$   & $134$                           & $5$  & $3$ &  No (best $d=135$)                  \\ \hline
\end{tabular}
\end{center} 
\end{table} 

Examples of the code $\C_{(q,\, m,\, \delta_3)}$ are listed in Table \ref{tab-bchthm4}. Some of them are optimal in 
the sense that they meet some bound on linear codes according to the Database. Some of them have the same 
parameters as the best codes known in the Database. When $(q,m)=(3,3)$, the code $\C_{(q,\, m,\, \delta_3)}$ 
has parameters $[26, 11, 8]$, which are the best possible according to \cite[p. 300]{Dingbk15}.  

\begin{theorem}\label{thm-dim3c}
Let $q>2$. 
The codes $\tilde{\C}_{(q,\, 3,\, q^3-q^2-q-2)}$ and $\C_{(q,\, 3,\, q^3-q^2-q-2)}$ have parameters 
$$ 
[q^3-1,\, 7,\, \tilde{d} \geq q^3-q^2-q-1] \mbox{ and } [q^3-1,\, 8,\, d \geq q^3-q^2-q-2],  
$$
respectively. 
\end{theorem}

\begin{proof}
When $m=3$, one can similarly prove that the third largest coset leader $\delta_3=\delta_2-1=q^3-q^2-q-2$ and 
$|C_{\delta_3}|=1$. The conclusions on the dimensions of $\tilde{\C}_{(q,\, 3,\, q^3-q^2-q-2)}$ and $\C_{(q,\, 3,\, q^3-q^2-q-2)}$ 
follow from Lemmas \ref{lem-1stcosetleader} and \ref{lem-2ndcosetleader}. The conclusions on the minimum 
distances follow from the BCH bound. 
\end{proof}

We conjecture that $\tilde{d} = q^3-q^2-q-1$ and $d = q^3-q^2-q-2$ for the two codes in Theorem \ref{thm-dim3c}, 
and invite the reader to settle this conjecture.

The following theorem follows from Theorem \ref{thm-dualdis181}, as  $\C_{(2, m, \delta_3)}^\perp$ is the even-weight 
subcode of $\tilde{\C}_{(2, m, \delta_3)}^\perp$. 

\begin{theorem}
The minimum distance $d^\perp$ of the dual of $\C_{(2, m, \delta_3)}$ is equal to $8$ when $m \geq 5$ 
is odd, and $6$ when $m \geq 6$ is even. 
\end{theorem}

\section{Summary and concluding remarks} 

The first contribution of this paper is the determination of the second and third largest $q$-cyclotomic coset leaders 
$\delta_2$ and $\delta_3$, which are documented in Lemmas \ref{lem-2ndcosetleader} and \ref{lem-3rdcosetleader}. 
The second contribution is the establishment of the dimensions, minimum distances and the weight distributions  
of the primitive BCH codes $\tilde{\C}_{(q,\, m,\, \delta_2)}$, $\tilde{\C}_{(q,\, m,\, \delta_3)}$, and the settlement  
of the dimensions and minimum distances of $\C_{(q,\, m,\, \delta_2)}$ and $\C_{(q,\, m,\, \delta_3)}$ for the case that 
$q$ is an odd prime. Sometimes a direct proof of the weight distribution formulas was given. Sometimes a bridge between 
these BCH codes and some cyclic codes with known weight distribution was established and employed to derive the weight 
distribution of the BCH codes. 

The third contribution of this paper is the characterisation of all the codewords in the code $\tilde{\C}_{(q,m,\delta_2)}$ 
with minimum weight documented in Theorem \ref{thm-minimumwt2}. This may be employed to derive $t$-designs with  
this code. It is noticed that characterising all the codewords with minimum weight in a linear code is a very difficult problem 
in general.  

Nonbinary Kasami codes were introduced and their weight distributions were settled in \cite{zengseta} and a class of cyclic 
codes were defined in \cite{Zeng10}. The fourth contribution of this paper is the proof of the fact that a subclass of the 
nonbinary Kasami codes and a subclass of cyclic codes treated in \cite{Zeng10} are actually equivalent to the BCH codes 
$\tilde{\C}_{(q,\, m,\, \delta_3)}$. The code $\tilde{\C}_{\delta_2}$ of (\ref{eqn-codeCdelta2}) was known to be equivalent 
to a cyclic code in the literature. Our contribution regarding this code is to prove that it is equivalent to the BCH code 
$\tilde{\C}_{(q,\, m,\, \delta_2)}$. 

A number of examples of the codes $\tilde{\C}_{(q,\, m,\, \delta_i)}$ and $\C_{(q,\, m,\, \delta_i)}$ 
for $i \in \{2, 3\}$ were worked out and put into several tables. Most of them are either optimal or almost 
optimal linear codes. It was also known that $\tilde{\C}_{(q,\, m,\, \delta_1)}$ and $\C_{(q,\, m,\, \delta_1)}$ 
are optimal with respect to the Griesmer bound \cite{Ding15}. The list of tables of best cyclic codes documented 
in \cite{Dingbk15} shows that BCH codes are among the best cyclic codes except in a few cases. Hence, it would 
be worthy to further investigate primitive BCH codes. 

To the best of our knowledge, the weight distributions of only a few classes of primitive BCH codes are determined 
in the literature. These codes are the following: 
\begin{enumerate}
\item The Reed-Solomon codes. 
\item The codes documented in this paper and two more subclasses of binary primitive BCH codes dealt with in \cite{Kasa69}.  
\end{enumerate}      
It is in general very difficult to determine the dimensions and minimum distances of BCH codes, let alone their weight distributions. 
For the weight distribution of cosets of some binary primitive BCH codes, the reader is referred to \cite{Charpin94,CharpinZeno} for information. 

Finally, we point out an application of some of the codes of this paper in secret sharing. Any linear code over
$\gf(q)$ can be employed to construct secret sharing schemes \cite{ADHK,CDY05,Mass93,YD06}. In order to have such
secret sharing scheme with interesting access structures, we need a linear code $\C$ over $\gf(q)$ such that
\begin{eqnarray}\label{eqn-sss}
\frac{w_{min}}{w_{max}}>\frac{q-1}{q},
\end{eqnarray}
where $w_{max}$ and $w_{min}$ denote the maximum and minimum nonzero weight in $\C$, respectively.

The codes $\tilde{\C}_{(q,\, m,\, \delta_2)}$ and $\tilde{\C}_{(q,\, m,\, \delta_3)}$ satisfy the inequality 
in (\ref{eqn-sss}) when $m \geq 5$, and can be employed to obtain secret sharing schemes with interesting 
access structures using the framework documented in \cite{ADHK,CDY05,Mass93,YD06}.

\section*{Acknowledgements}

The authors would thank Dr. Pascale Charpin for providing information on known results on narrow-sense primitive BCH codes.

\end{document}